\pdfoutput = 1
\documentclass{sig-alternate-10pt}

\usepackage{amsmath,amsfonts,amssymb,graphicx,graphics,epsf}
\usepackage{mathrsfs}
\usepackage[ansinew]{inputenc}
\usepackage{fancyhdr}
\usepackage{color}
\newtheorem{theorem}{Theorem}[section]

\newtheorem{lemma}[theorem]{Lemma}

\newcommand{\e}{\varepsilon}

\newcommand\Sq{\mathscr{S}}

\newcommand{\arr}[1]{\overrightarrow{#1}}

\begin{document}

\numberofauthors{2}

\title{Undirected Connectivity of Sparse Yao Graphs}
%\author{
%Mirela Damian and Abhaykumar Kumbhar\\
%\affaddr{Department of Computer Science, Villanova University, Villanova, PA 19085, USA} \\
%\email{\{mirela.damian,abhaykumar.kumbhar\}@villanova.edu}
%}

\author{
\alignauthor
Mirela Damian \\
\affaddr{Department of Computer Science, Villanova University, Villanova, PA 19085, USA} \\
\email{mirela.damian@villanova.edu}
\alignauthor
Abhaykumar Kumbhar \\
\affaddr{Department of Computer Science, Villanova University, Villanova, PA 19085, USA} \\
\email{abhaykumar.kumbhar@villanova.edu}
}

\maketitle

%\{Department of Computer Science, Villanova University, Villanova, PA 19085, USA}
%            \email{\{mirela.damian,abhaykumar.kumbhar\}@villanova.edu}

\begin{abstract}
Given a finite set $S$ of points in the plane and a real value $d > 0$, the
$d-$radius disk graph $G^d$ contains all edges connecting pairs of
points in $S$ that are within distance $d$ of each other.
For a given graph $G$ with vertex set $S$, the Yao subgraph $Y_k[G]$ with
integer parameter $k > 0$ contains, for each point $p \in S$,  a shortest
edge $pq \in G$ (if any) in each of the $k$ sectors defined by $k$
equally-spaced rays with origin $p$.
Motivated by communication issues in mobile networks with directional
antennas, we study the connectivity properties of $Y_k[G^d]$, for
small values of $k$ and $d$. In particular, we derive lower and upper
bounds on the minimum radius $d$ that renders $Y_k[G^d]$
connected, relative to the unit radius assumed to render $G^d$ connected.
We show that $d = \sqrt{2}$ is necessary and sufficient for the
connectivity of $Y_4[G^d]$. We also show that, for
$d \le 5-\frac{2}{3}\sqrt{35}$, the graph $Y_3[G^d]$ can be disconnected,
but $Y_3[G^{2/\sqrt{3}}]$ is always connected.
Finally, we show that $Y_2[G^d]$ can be disconnected, for any $d \ge 1$.
\end{abstract}

\section{Introduction}
Let $S$ be a finite set of points in the plane and let $G = (S, E)$
be an arbitrary (undirected) graph with node set $S$.
The \emph{directed Yao graph} $\arr{Y_k}[G]$ with integer parameter
$k > 0$ is a subgraph of $G$ defined as follows.
At each point $p \in S$, $k$ equally-spaced rays with origin $p$
define $k$ cones. %~\cite{Yao82}
In each cone, pick a shortest edge $pq$ from $G$, if any, and add the
directed edge $\arr{pq}$ to $\arr{Y_k}[G]$. Ties are broken arbitrarily.
The \emph{undirected} Yao graph $Y_k[G]$ ignores the directions of edges,
and includes an edge $pq$ if and only if either $\arr{pq}$ or $\arr{qp}$
is in $\arr{Y_k}[G]$.

For a fixed real value $d > 0$, let $G^d(S)$ denote the $d-$radius disk graph
with node set $S$, in which two nodes $p, q \in S$ are adjacent if and only if
$|pq| \le d$. Most often we will refer to $G^d(S)$ simply as $G^d$, unless
the point set $S$ that defines $G^d$ is unclear from the context.
Under this definition, $G^1$ is the unit disk graph
(UDG), and $G^\infty$ is the complete Euclidean graph, in which any two points
are connected by an edge.
In this paper we study the connectivity of the undirected Yao graph $Y_k[G^d]$,
for small values $k \in \{2, 3, 4\}$ and $d \ge 1$.
Underlying our study is the assumption that $G^1$ is connected.
(For example, $G^1$ can be thought of as the graph connecting
all pairs of points that are within distance no greater than the length of
the bottleneck edge in a minimum spanning tree for $S$, normalized to
one.) In this context, we investigate the following problem:

\medskip
\noindent
\begin{center}
\begin{minipage}{0.9\linewidth}
\emph{Let $S$ be an arbitrary set of points in the plane, and suppose that the
unit radius graph $G^1$ defined on $S$ is connected.
What is the smallest real value $d \ge 1$ for which $Y_k[G^d]$ is connected?}
\end{minipage}
\end{center}

\medskip
\noindent
Throughout the paper, we will refer to the minimum value $d$ that renders
$Y_k[G^d]$ connected as the \emph{connectivity radius} of $Y_k[G^d]$.

\medskip
\noindent
Our research is inspired by the use of wireless directional antennas
in building communication networks. Unlike an omnidirectional antenna, which
transmits energy in all directions, a directional antenna can concentrate its
transmission energy within a narrow cone; the narrower the cone, the longer
the transmission range, for a fixed transmission power level. Directional
antennas are preferable over omnidirectional antennas, because they reduce
interference and extend network lifetime, two criteria of utmost importance
in wireless networks operating on scarce battery resources.

Directed Yao edges can be realized with narrow directional antennas (otherwise called
laser-beam antennas, to imply a small cone angle, close to zero).
One attractive property of Yao graphs is that they can be efficiently constructed locally,
because each node can select its incident edges based on the information from nodes in
its immediate neighborhood only. This enables each node to repair the communication structure
quickly in the face of dynamic and kinetic changes, providing strong support for
node mobility.

The limited number of antennas per node ($1$ to $4$ in practice), raises the fundamental
question of connectivity of Yao graphs $Y_k$, for small values of $k$. If the
communication graph induced by antennas operating in omnidirectional
mode is connected, by how much must an antenna radius
increase to guarantee that $k$ laser-beam antennas at each node, pointing
in the direction of the $Y_k$ edges, preserve connectivity?
In this paper we focus our attention on small $k$ values
($2$, $3$ and $4$) corresponding to the number of antennas commonly used in practice.

\vspace{0.4em}
\subsection{Prior Results}
Yao graphs have been extensively studied in the area of computational geometry,
and have been used in constructing efficient wireless communication
networks~\cite{LWW01,li02sparse,SWLF04,LSW05}. Applying the Yao structure
on top of a dense communication graph, in order to obtain a sparser graph,
is a very natural idea. Most existing results concern Yao graphs
$Y_k[G^\infty]$ with $k \ge 6$. These graphs exhibit nice spanning properties,
in the sense that the length of a shortest path between any two nodes
$p, q \in Y_k[G^\infty]$ is only a constant times the Euclidean distance
$|pq|$ separating $p$ and $q$~\cite{bmnsz-agbsp-03,BDK+10,DR10}.
In the context of using laser-beam antennas to realize $Y_k$ however, these
results could only be applied if $6$ or more antennas were available at each node,
which is a rather impractical requirement.
Few results exist on Yao graphs $Y_k$, for small values of $k$ (below $6$).
It has been shown that $Y_2[G^\infty]$ and $Y_3[G^\infty]$ are not
spanners~\cite{MollaThesis09}, and that $Y_4[G^\infty]$ is a
spanner~\cite{BDK+10}.
However, as far as we know, no results exist on $Y_k[G^d]$, for
any fixed radius $d \ge 1$.

\vspace{0.4em}
\subsection{Our Results}
We develop lower and upper bounds on the connectivity
radius of $Y_2$, $Y_3$ and $Y_4$, relative to the unit radius. (Recall
that our assumption that the unit radius disk graph $G^1$ is connected.)
We prove tight lower and upper bounds equal to $\sqrt{2} \approx 1.414$ on the
connectivity radius $d$ of $Y_4[G^d]$. Surprisingly, we prove a smaller
upper bound equal to $\frac{2}{\sqrt{3}} \approx 1.155$ on the connectivity radius
$d$ of $Y_3[G^d]$.
This is somewhat counterintuitive, as one would expect that fewer outgoing
edges per node ($3$ in the case of $Y_3$, compared to $4$ in the case of $Y_4$)
would necessitate a higher connectivity radius, however our results show that
this is not always the case.
We also derive a lower bound of
$5-\frac{2}{3}\sqrt{35} \approx 1.056$ on the connectivity radius $d$
of $Y_4[G^d]$, leaving a tiny interval $[1.056, 1.155]$ on which the
connectivity of $Y_4$ remains uncertain.
Finally, we show that $Y_2[G^d]$ can be disconnected, for any fixed
value $d \ge 1$.

\vspace{0.4em}
\subsection{Definitions}
Let $S$ be a fixed set of points in the plane.
At each node $p \in S$, let $r_1, r_2, \ldots, r_k$ denote
the $k$ rays originating at $p$, with $r_1$ horizontal along the $+x$ axis
(see Figure~\ref{fig:y3defs}, for $k = 3$).
Let $C_i(p)$ to denote the half-open cone delimited by $r_i$ and $r_{i+1}$,
including $r_i$ but excluding $r_{i+1}$.
(Here we use $r_{k+1}$ to mean $r_1$.)
%
%%%%%%%%%%%%%%%%%%%%%%%%%%%%%%%%%Figure Begin
\begin{figure}[htpb]
\centering
\includegraphics[width=0.5\linewidth]{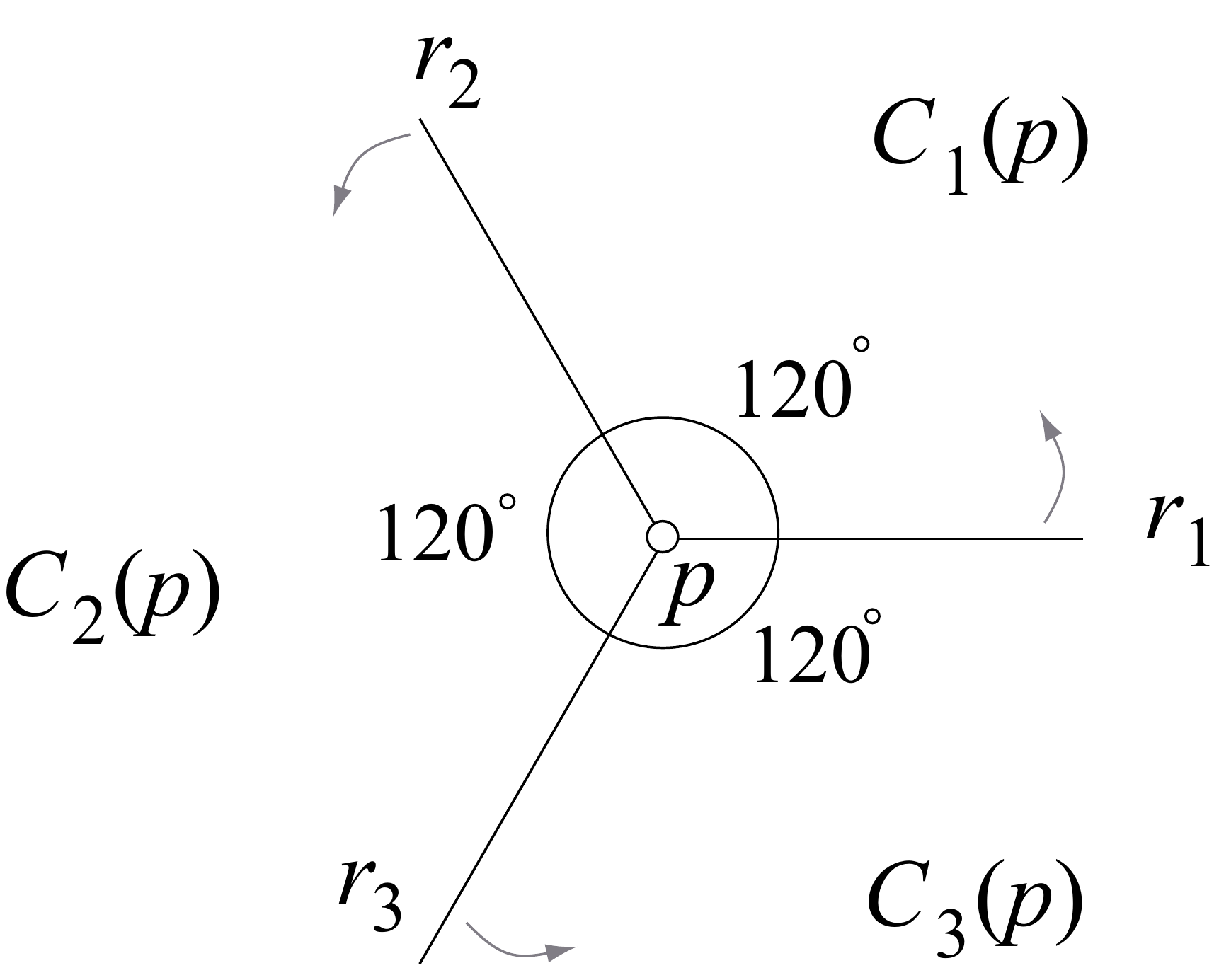}
\caption{Rays and (half-open, half-closed) cones used in constructing $Y_3$.}
\label{fig:y3defs}
\end{figure}
%%%%%%%%%%%%%%%%%%%%%%%%%%%%%%%%%Figure End
%
For any point $p \in S$, let $x(p)$ denote
the $x-$coordinate of $p$ and $y(p)$ denote the $y-$coordinate of $p$. For any
$p, q \in S$, let $|pq|$ denote the Euclidean distance between $p$ and $q$.
For any point $p \in S$ and any real value $\delta > 0$, let $D(p, \delta)$ be
the closed disk with center $p$ and radius $\delta$.

\section{Connectivity of $Y_4$}
\label{sec:y4}
In this section we derive tight lower and upper bounds on the
connectivity radius $d$ for $Y_4[G^d]$. Recall that our work
relies on the assumption that $G^1$ is connected.

\begin{theorem}
There exist point sets $S$ with the property that $G^1(S)$ is connected,
but $Y_4[G^d]$ is disconnected, for any $1 \le d < \sqrt{2}$.
\label{thm:lb4}
\end{theorem}
\begin{proof}
We construct a point set $S$ that meets the conditions of the theorem.
Note that $d < \sqrt{2}$ implies that $1-\sqrt{d^2-1} > 0$, meaning that
there exists a real value $\e$ such that $0 < \e < 1-\sqrt{d^2-1}$,
%This $\e$ value satisfies the inequality
which is equivalent to
\[
    1 + (1-\e)^2 > d^2
\]
Let $p$ and $q$ be the endpoints of a vertical segment of length $1$,
with $p$ below $q$. In Figure~\ref{fig:y2disconnect} the segment
$pq$ is shown slightly slanted to the left, merely to reinforce
our convention that $pq \in C_2(p)$ and $qp \in C_4(q)$.
%%%%%%%%%%%%%%%%%%%%%%%%%%%%%%%%%Figure Begin
\begin{figure}[hptb]
\centering
\includegraphics[width=\linewidth]{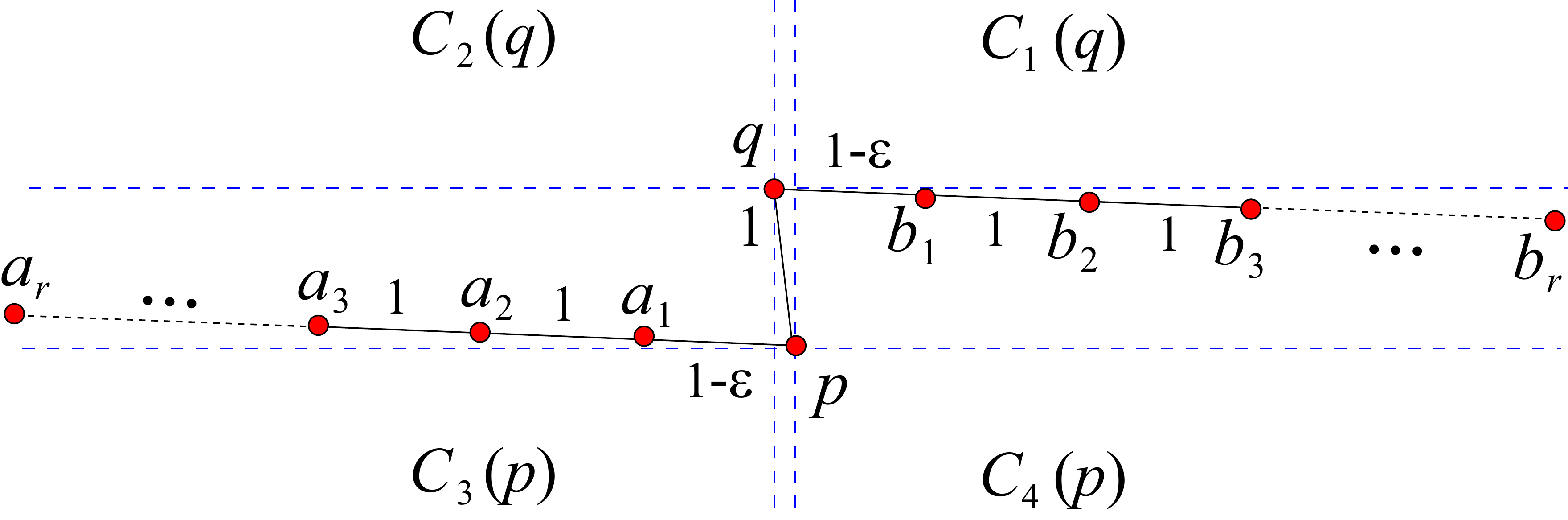}
\caption{Point set $S$ and $G^1(S) \equiv G^d(S)$, with $d < \sqrt{2}$; $Y_4[G^d]$ is disconnected.}
\label{fig:y2disconnect}
\end{figure}
%%%%%%%%%%%%%%%%%%%%%%%%%%%%%%%%%Figure End
Shoot a horizontal ray from $p$ leftward, then slightly rotate it clockwise
about $p$ by a tiny angle $\alpha$, so that the ray lies entirely in $C_2(p)$.
Distribute points
$a_1, a_2, \ldots, a_r$ in this order along this ray such that $|pa_1| = 1-\e$,
and $|a_ia_{i+1}| = 1$, for each $i$. Let $b_i$ be the point symmetric to $a_i$ with
respect to the midpoint of $pq$. Let
\[S = \{p, q, a_i, b_i ~|~ 1 \le i \le r\}.
\]

In the limit, as $\alpha$ approaches $0$,
the angle $\angle a_1pq$ approaches $\pi/2$ and $|a_1q| = \sqrt{1 + (1-\e)^2} > d$.
This means that $a_1q$ is not an edge in $G^d$. Because
$|a_ib_j| > |a_i q| \ge |a_1q| > d$ for each $i, j \ge 1$, we have that no $a-$point
is directly connected to a $b-$point in $G^d$. It follows that the graph $G^d$
is a path (depicted in Figure~\ref{fig:y2disconnect}).

We now show that $pq \not\in Y_4[G^d]$, which along with the fact that
$G^d$ is a path, yields that claim that $Y_4[G^d]$ is disconnected.
First note that $\arr{pq}$ is not an edge in $Y_4[G^d]$. This is
because $a_1$ is in the same cone $C_2(p)$ as $q$, and $|pa_1| = 1-\e < 1 = |pq|$.
Similarly, $\arr{qp}$ is not an edge in $Y_4[G^d]$, because
$b_1$ is in the same cone $C_4(q)$ as $p$, and $|qb_1| = 1-\e < 1 = |qp|$.
We conclude that $Y_4[G^d]$ is disconnected.
\end{proof}

\paragraph{Upper Bound $d = \sqrt{2}$}
We now show that $Y_4[G^d]$ is always connected for $d = \sqrt{2}$,
matching the lower bound from Theorem~\ref{thm:lb4}.
First we introduce a few definitions.
%
%%%%%%%%%%%%%%%%%%%%%%%%%%%%%%%%%Figure Begin
\begin{figure}[htpb]
\centering
\includegraphics[width=\linewidth]{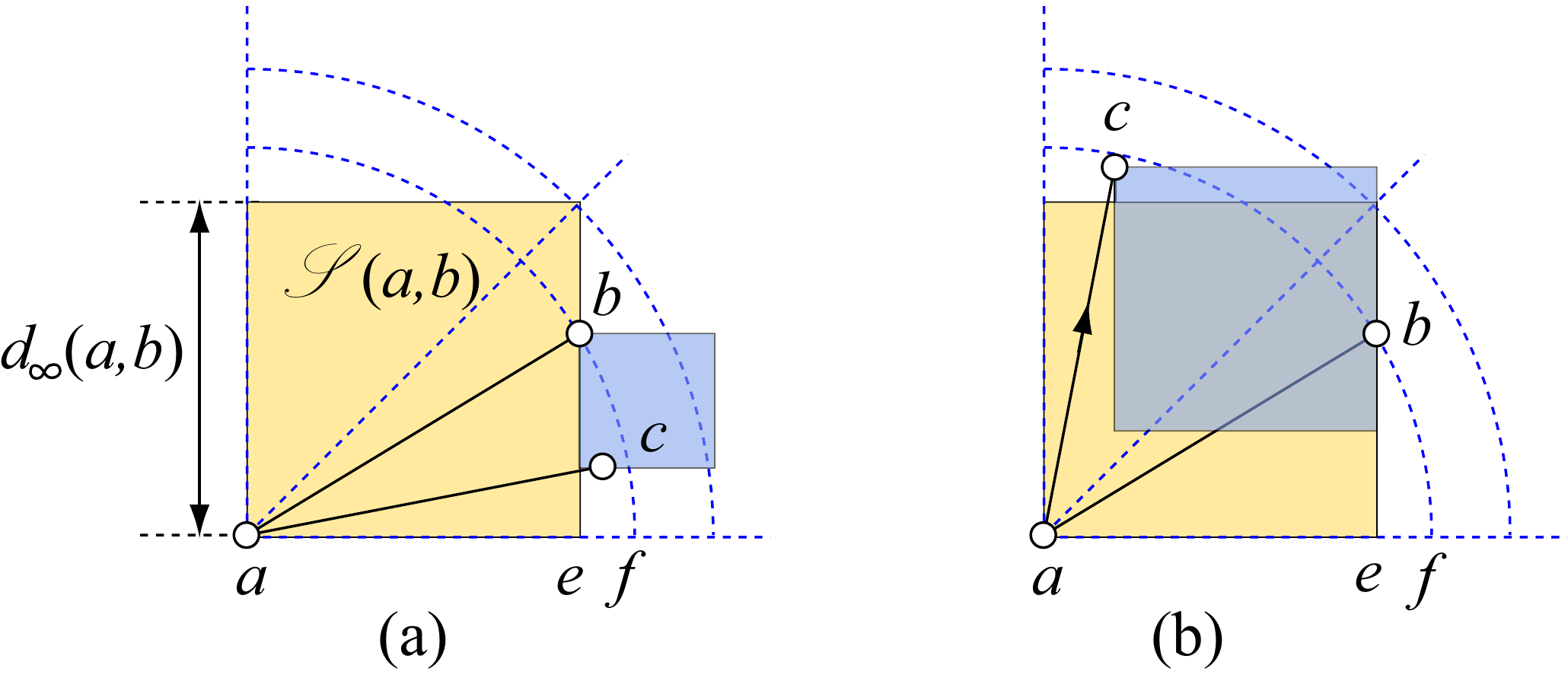}
\caption{Theorem~\ref{thm:y4connect}: $d_\infty(b, c) < d_\infty(a, b)$ (a)
$b, c$ lie on the same side of the bisector (b) $b, c$ lie on opposite sides
of the bisector.}
\label{fig:y4connect}
\end{figure}
%%%%%%%%%%%%%%%%%%%%%%%%%%%%%%%%%Figure End
%
For any pair of points $a, b$, let $d_\infty(a,b)$ denote the
$L_\infty$ distance between $a$ and $b$, defined as
\[ d_\infty(a, b) = \max\{|x(a)-x(b)|, |y(a)-y(b)|\}
\]
Let $\Sq(a, b)$ be the square with corner $a$ whose boundary contains $b$,
of side length $d_\infty(a, b)$ (see Figure~\ref{fig:y4connect}a).
The following inequalities follow immediately from the fact that $ab$ is a
line segment inside $\Sq(a, b)$:
\begin{equation}
d_\infty(a, b) \le |ab| \le d_\infty(a,b)\sqrt{2}
\label{eq:inf}
\end{equation}

\begin{theorem}
For any point set $S$ such that $G^1(S)$ is connected, $Y_4[G^{\sqrt{2}}]$ is also
connected.
\label{thm:y4connect}
\end{theorem}
\begin{proof}
The proof is by contradiction.
Assume to the contrary that $G^1$ is connected, but
$Y_4[G^{\sqrt{2}}]$ is disconnected. Then $Y_4[G^{\sqrt{2}}]$ has at least two
connected components, say $J_1$ and $J_2$.
Since $G^1 \subseteq G^{\sqrt{2}}$ is connected,
there is an edge $pq \in G^1$, with $p \in J_1$ and $q \in J_2$.
To derive a contradiction, consider two points $a, b \in S$, with
$a \in J_1$ and $b \in J_2$, that minimize $d_\infty(a, b)$. Then
\begin{eqnarray*}
d_\infty(a, b) & \le & d_\infty(p, q) ~~~\mbox{(by choice of $ab$)}\\
               & \le & |pq| ~~~~~~~~~\mbox{(by ~(\ref{eq:inf}))} \\
               & \le & 1  ~~~~~~~~~~~~\mbox{(because $pq \in G^1$)}
\end{eqnarray*}
This along with the second inequality from~(\ref{eq:inf}) implies
$|ab| \le \sqrt{2}$, therefore $ab \in G^{\sqrt{2}}$.
To simplify our analysis, rotate $S$ so that $b$ lies in the
lower half of $C_1(a)$. % (excluding the bisector).

If $ab \in Y_4[G^{\sqrt{2}}]$, then $ab$ connects $J_1$ and $J_2$,
contradicting our assumption that $J_1$ and $J_2$ are disjoint connected
components. So $ab \not\in Y_4[G^{\sqrt{2}}]$. However $ab \in G^{\sqrt{2}}$
and $b \in C_1(a)$, therefore
there is $\arr{ac} \in Y_4[G^{\sqrt{2}}]$, with
$c \in C_1(a)$ and $|ac| \le |ab|$. If $c$ lies inside $\Sq(a, b)$, then
$d_\infty(b, c) < d_\infty(a, b)$, because each of the horizontal and
vertical distance between $b$ and $c$ is strictly smaller than the
side length of $\Sq(a, b)$. This along with the fact that
$bc$ connects $J_1$ and $J_2$, contradicts our choice of $ab$.
So $c$ must lie outside of $\Sq(a, b)$ (but not outside of $D(a, |ab|)$,
because $|ac| \le |ab|$).

Let $e$ be the lower right corner of $\Sq(a, b)$, and let $f$ be intersection
point between the boundary of $D(a, |ab|)$ and the horizontal ray through $a$
in the direction of $e$ (see Figure~\ref{fig:y4connect}a).
We will be using the fact that
\begin{equation}
|be| > |ef|
\label{eq:circle}
\end{equation}
(This follows from the fact that $\angle bfe = \angle fba > \angle fbe$,
and the Law of Sines applied on $\triangle bef$.)

We now derive a contradiction to our choice of $ab$ as follows.
If both $b$ and $c$ lie in the lower half of $C_1(a)$,
as depicted in Figure~\ref{fig:y4connect}a, then
$|y(b)-y(c)| < d_\infty(a, b)$. Also $|x(b)-x(c)| < |ef|$, which
by inequality~(\ref{eq:circle}) is no longer than $d_\infty(a, b)$.
It follows that $d_\infty(b, c) < d_\infty (a, b)$, which
along with the fact that $bc$ connects $J_1$ and $J_2$,
contradicts our choice of $ab$.
If $b$ and $c$ lie on opposite sides of the bisector of $C_1(a)$,
as depicted in Figure~\ref{fig:y4connect}b, then
the vertical distance from $c$ to the top side of $\Sq(a, b)$ is smaller
than $|ef|$, which in turn is smaller than $|be|$
(by inequality~(\ref{eq:circle})). It follows that
$|y(c)- y(b)| < d_\infty(a, b)$. Also, because $c$ lies strictly to the
right of $a$, we have that $|x(c)- x(b)| < d_\infty(a, b)$. These together
show that $d_\infty(b,c) < d_\infty(a, b)$. This  along with the
fact that $bc$ connects $J_1$ and $J_2$, contradicts our choice of $ab$.
\end{proof}
Theorems~\ref{thm:lb4} and~\ref{thm:y4connect} together establish
matching lower and upper bounds (equal to $\sqrt{2}$) for the
connectivity radius of $Y_4$.

\section{Connectivity of $Y_3$}
The Yao graph $Y_3$ has three outgoing edges per node, compared to
four outgoing edges in the case of $Y_4$. So one would expect that
the radius necessary to maintain $Y_3$ connected would exceed the radius
necessary to maintain $Y_4$ connected. However, our results show that
an antenna radius equal to $\frac{2}{\sqrt{3}} < \sqrt{2}$ suffices to maintain $Y_3$
connected. This is a surprising result, given that a radius of
$\sqrt{2}$ is necessary and sufficient to maintain $Y_4$ connected,
as established in the previous section.

\begin{theorem}
There exist point sets $S$ with the property that $G^1(S)$ is connected,
but $Y_3[G^d]$ is disconnected, for any $1 \le d < 5-\frac{2}{3}\sqrt{35}$.
\label{thm:y3lb}
\end{theorem}
\begin{proof}
We construct a point set $S$ that satisfies the conditions of the theorem.
Start with an isosceles trapezoid $pa_1b_1q$ of unit altitude and bases $pq$ and
$a_1b_1$, with $|pq| = 1$ and $|a_1b_1| = 1+\e$, for some small real value
$0 < \e < 1$, to be determined later. Place a point
$x$ on $pa_1$ at distance $|pa_1|/3$ from $p$, and a second point
$y$ on $qb_1$ at distance $|qb_1|/3$ from $b_1$. Then simply reflect
$px$ about the vertical line through $p$, and $qy$ about the vertical
line through $q$. As we will later see, this places $px$ and $pq$ in the
same cone of $p$ (after a $90^\circ$ counterclockwise rotation), so that
$px$ and $pq$ compete in the edge selection process at $p$.

The result is the shaded polygon depicted in Figure~\ref{fig:y3disconnect}a.
%
%%%%%%%%%%%%%%%%%%%%%%%%%%%%%%%%%Figure Begin
\begin{figure}[hptb]
\centering
\includegraphics[width=\linewidth]{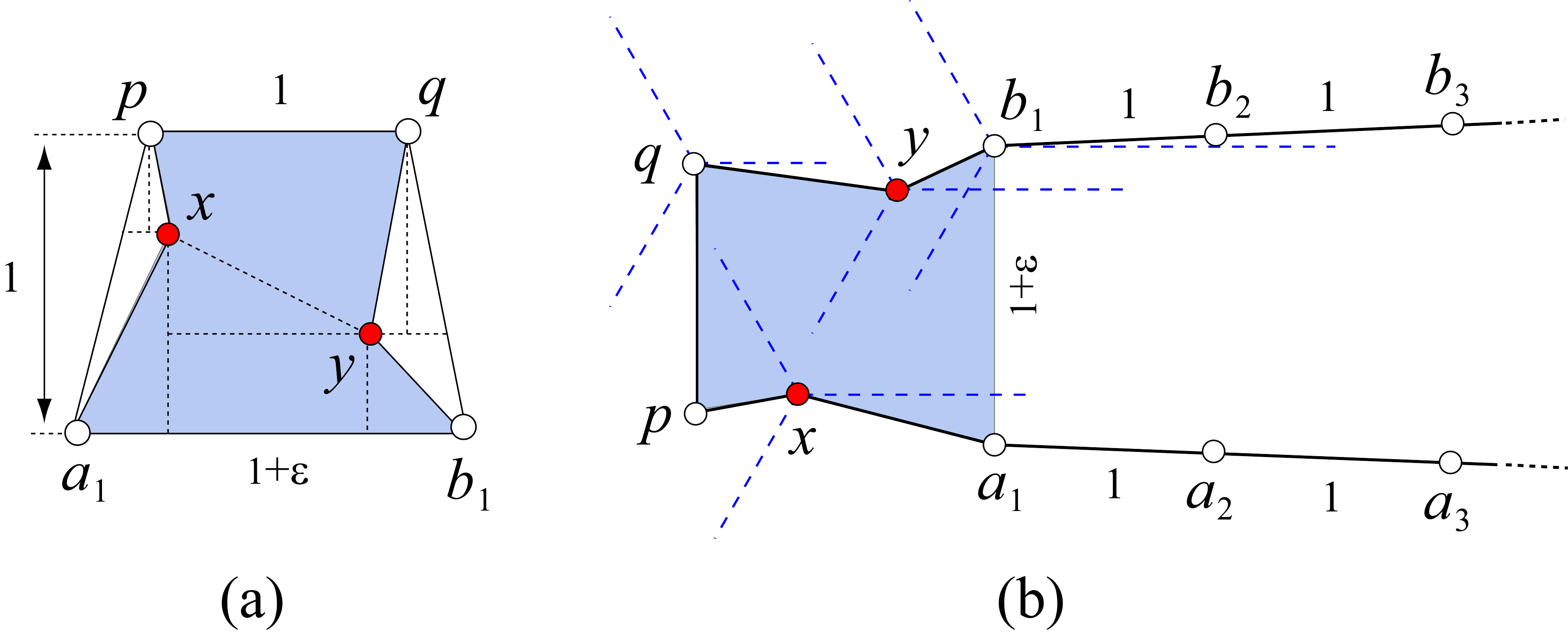}
\caption{(a) Construction of $pxa_1b_1yq$ (b) Point set $S$
and $G^d(S)$, with $1 \le d < 10/9$.}
\label{fig:y3disconnect}
\end{figure}
%%%%%%%%%%%%%%%%%%%%%%%%%%%%%%%%%Figure End
%
Simple calculations show that the vertical distance between $x$ and $y$ is $1/3$, and the
horizontal distance between $x$ and $y$ is
\[
1-\frac{1}{3}\cdot\frac{\e}{2} - \frac{2}{3}\cdot\frac{\e}{2} = 1 -\frac{\e}{2}
\]
It follows that $|xy|^2 = \left(\frac{1}{3}\right)^2+\left(1-\frac{\e}{2}\right)^2
= 1 + \frac{9\e^2-36\e+4}{36}$. We will later
require that $|xy| > d$ and $|a_1b_1| > d$, so that neither $xy$ nor $a_1b_1$ is
a candidate for $Y_3[G^d]$. These two inequalities reduce to
\[
\begin{cases}
  9\e^2-36\e+40-36d > 0 \\
  1+\e > d
\end{cases}
\]
Simple calculations yield the solution
\begin{align}
1 \le d &<   5-\frac{2}{3}\sqrt{35} \\
d-1 < \e &<  2 - \frac{2}{3}\sqrt{9d-1}\label{eq:de}.
\end{align}

\noindent
By the triangle inequality, $|xa_1| < 2/3 + \e/2$.
It can be easily verified that the above constraints on $\e$ and $d$
yield $|xa_1| < 1$. Similarly, each of $px$, $qy$ and $yb_1$ has
length less than $1$. Also note that
\[|xq| > |xy| > d,\]
since the horizontal distance between $x$ and $q$ is greater than the
horizontal distance between $x$ and $y$, and the vertical distance is
$1/3$ in both cases. Similarly, $|a_1y| > |xy| > d$.

We are now ready to construct $S$.
Start by rotating the polygon $pxa_1b_1yq$ counterclockwise by $90^\circ$, so that
it lies  on its side, as in Figure~\ref{fig:y3disconnect}b.
Shoot a horizontal ray rightward from $a_1$, then rotate it slightly clockwise
so that it lies entirely in $C_3(a_1)$.
Distribute points $a_2, a_3, \ldots, a_r$ at unit intervals along this ray.
Let $b_i$ be the reflection of $a_i$ with respect to the horizontal through
the midpoint of $pq$, for each $i > 1$. Our point set is
\[
S = \{p, q, x, y, a_1, a_2, \ldots, a_r, b_1, b_2, \ldots, b_r\}.
\]
The graph $G^1$ is a path (depicted in Figure~\ref{fig:y3disconnect}b)
and is therefore connected. We now show that $Y_3[G^d]$ is disconnected.

By construction, the following inequalities hold:
$|a_1b_1|$ $>$ $d$; $|xq| > |xy| > d$; $|a_1q| > |a_1y| > |xy| > d$; %$|pa_1| > 1$; $|pb_1| > 1$;
and $|a_ib_{i+j}| > |a_ib_i| \ge |a_1b_1| > d$, for any
$i \ge 1$ and any $j \ge 0$ (because $\angle a_ib_ib_{i+j}$ is obtuse).
By symmetry, similar arguments hold for the $b-$points as well.
It follows that the graph $G^d$ is a path identical to $G^1$,
therefore the removal of any edge from $G^d$ disconnects it.

Next we show that $pq \not\in Y_3[G^d]$, which along with the observation above
implies that $Y_3[G^d]$ is disconnected.
By construction, $|px| < |pq|$. This
along with the fact that both $x$ and $q$ lie in the same cone
$C_1(p)$, implies that $p$ does not select $pq$ for inclusion
in $Y_3[G^d]$.
Similarly, $|qy| < |qp|$. This along with the fact that
both $y$ and $p$ lie in the same cone $C_3(q)$, implies that
$q$ does not select $qp$ for inclusion in $Y_3[G^d]$. These together show that
$pq \not\in Y_3[G^d]$, therefore $Y_3[G^1]$ is disconnected.
\end{proof}

\paragraph{Upper Bound $d \le 2/\sqrt{3}$}
Next we derive an upper bound on the connectivity radius for $Y_3$.
The approach adopted here is somewhat similar to
the one employed in the proof of Theorem~\ref{thm:y4connect}, but it uses a generalized
distance function $d_R$ (in place of $d_\infty$), to measure the distance
between two connected components.
%
%%%%%%%%%%%%%%%%%%%%%%%%%%%%%%%%%Figure Begin
\begin{figure}[hp]
\centering
\includegraphics[width=0.5\linewidth]{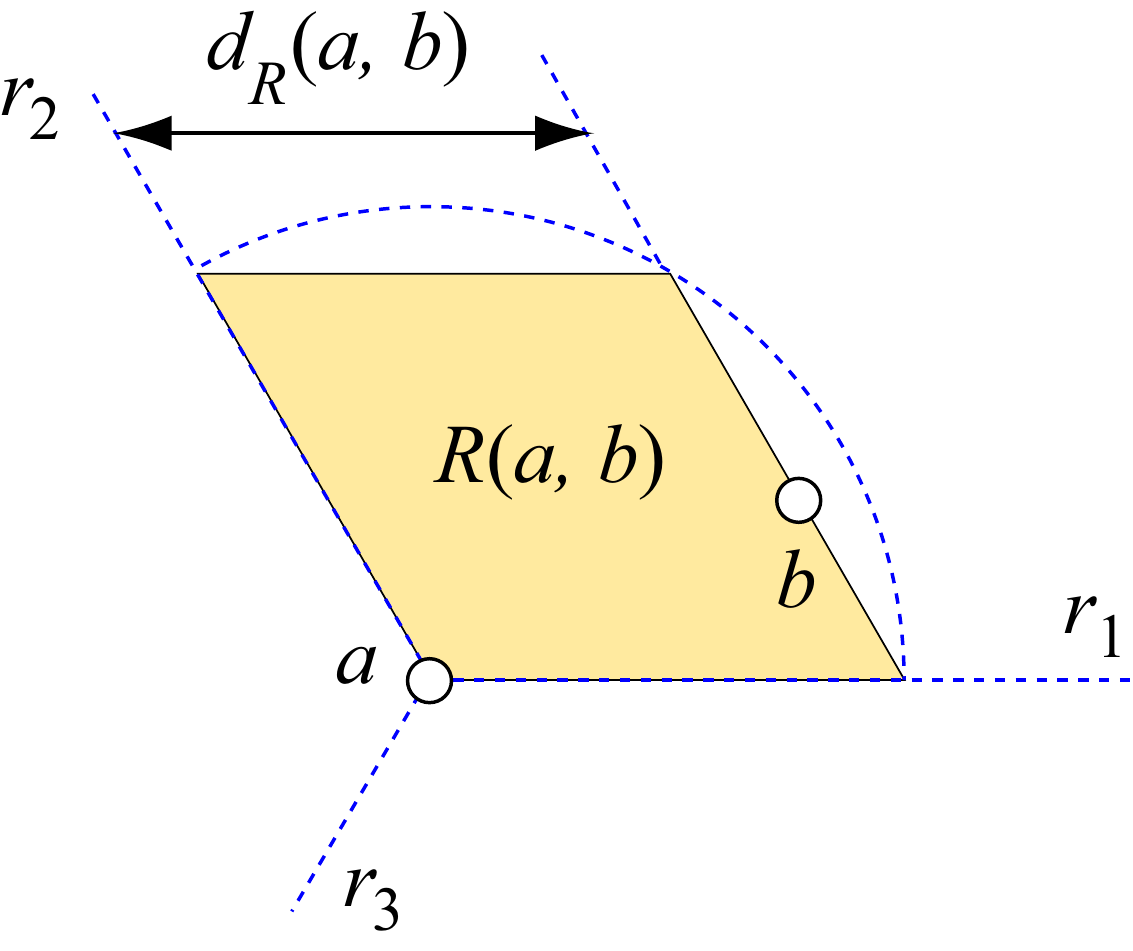}
\caption{Rhombus $R(a, b)$ of side length $d_R(a, b)$.}
%(b) $d_R(a, c) < d_R(a, b)$.}
\label{fig:Rdef}
\end{figure}
%%%%%%%%%%%%%%%%%%%%%%%%%%%%%%%%%Figure End
%
We define $d_R$ as follows. For any point $a \in S$ and any point $b \in C_i(a)$,
let $R(a, b)$ denote the closed rhombus with corner $a$ and edges parallel to $r_i$
and $r_{i+1}$, whose \emph{boundary} $\partial R(a, b)$ contains $b$ (see Figure~\ref{fig:Rdef}).
(Recall that $C_i(a)$ is the half-open cone with apex $a$ that includes
$r_i$ and excludes $r_{i+1}$.)
Define $d_R(a, b)$ to be the side length of $R(a, b)$. Clearly, $d_R(a,a) = 0$.
Because our approach does not use the triangle inequality on $d_R$, we skip the proof that
$d_R$ is a distance metric, and focus instead on the symmetry property of $d_R$
(Property (i) of Lemma~\ref{lem:dineq} below), and the relationship between $d_R$
and the Euclidean distance.

\noindent
\begin{lemma}
For any pair of points $a, b \in S$ the following properties hold:
\begin{enumerate}
\item[(i)] $d_R(a, b) = d_R(b, a)$.
\item[(ii)] $|ab| \le d_R(a, b)$.
\item[(iii)] $|ab| \ge d_R(a, b) \frac{\sqrt{3}}{2}$.
\end{enumerate}
\label{lem:dineq}
\end{lemma}
\begin{proof}
To simplify our analysis, rotate $S$ so that $b \in C_1(a)$, as
depicted in Figure~\ref{fig:dineq}.
Consider the quadrilateral $bcef$ from Figure~\ref{fig:dineq}a,
with sides $ce \in \partial R(a, b)$ and $bf \in \partial R(b, a)$;
$bc$ and $ef$ are parallel, since they are both parallel to $r_2$;
and $\angle cef$ and $\angle bfe$ are each $60^\circ$. These together show
that $bcef$ is an isosceles trapezoid, meaning that $|ce| = |bf|$. Since
$|ce| = d_R(a, b)$ and $|bf| = d_R(b, a)$, Property (i) holds.

%
%%%%%%%%%%%%%%%%%%%%%%%%%%%%%%%%%Figure Begin
\begin{figure}[htpb]
\centering
\includegraphics[width=\linewidth]{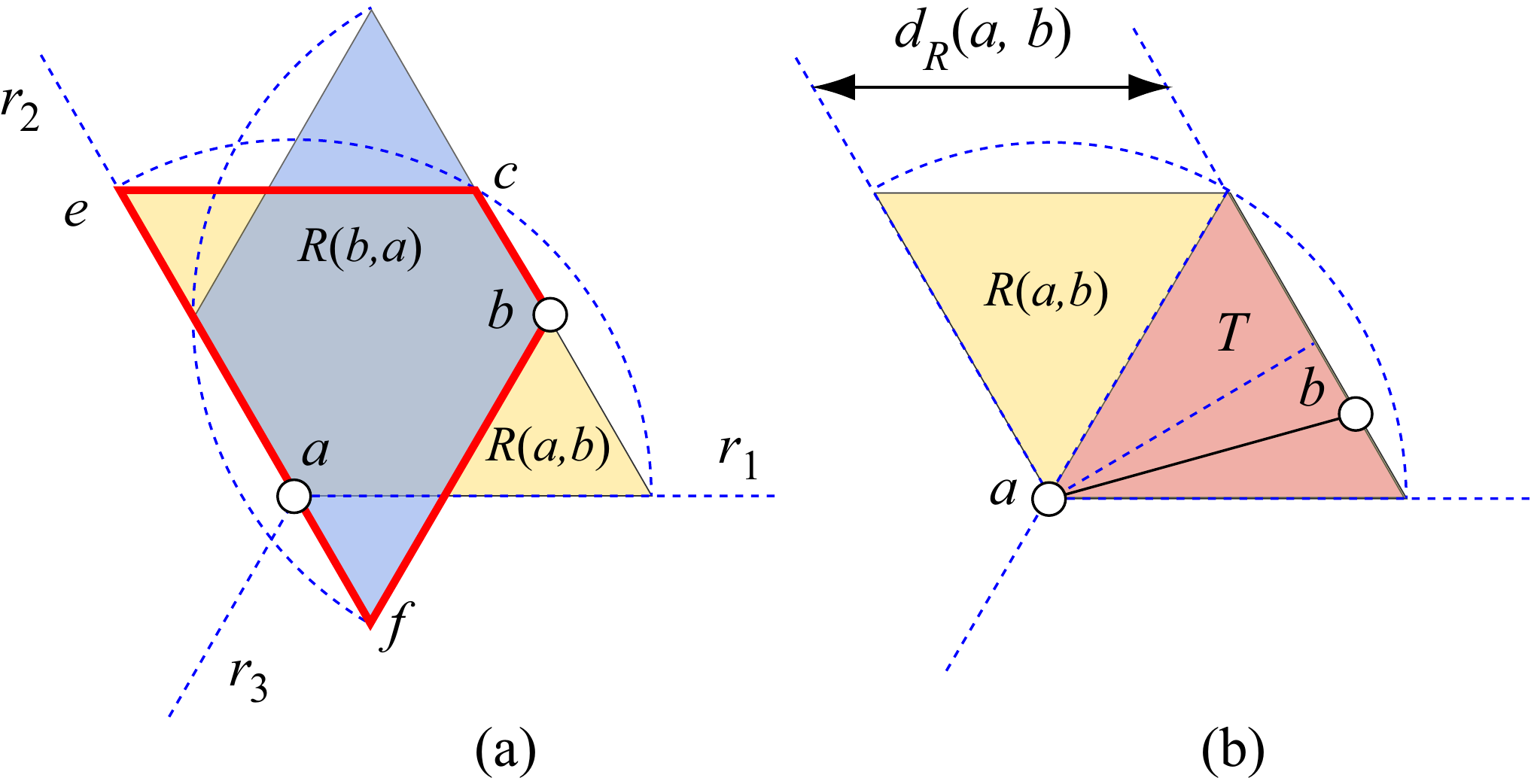}
\caption{Lemma~\ref{lem:dineq}: (a) $d_R(a, b) = d_R(b, a)$ (b) Relationship between $d_R(a, b)$ and $|ab|$.}
\label{fig:dineq}
\end{figure}
%%%%%%%%%%%%%%%%%%%%%%%%%%%%%%%%%Figure End
%
Now note that $R(a, b)$ is the union of two equilateral triangles of side length $d_R(a, b)$,
adjacent alongside the bisector of $C_1(a)$. Also
note that $ab$ is a segment that connects $a$ to the opposite side
in one of these triangles -- call it $T$.
It follows that $ab$ is no longer than the side of $T$ (thus yielding inequality (ii)),
and no shorter than the height of $T$ (thus yielding inequality (iii)). This completes the proof.
\end{proof}

%%%%%%%%%%%%%%%%%%%%%%%%%%%%%%%%%Figure Begin
%\begin{figure}[htpb]
%\centering
%\includegraphics[width=0.26\linewidth]{d-property}
%\caption{Lemma~\ref{lem:outR}.}
%\label{fig:outR}
%\end{figure}
%%%%%%%%%%%%%%%%%%%%%%%%%%%%%%%%%Figure End

%\begin{lemma}
%Let $a, b, c \in S$ be such that $b$ and $c$ lie in the same cone with apex $a$,
%and $d_R(a, b) \le d_R(a, c)$. Then the following two properties hold: (i)
%$c$ lies on $\partial R(a, b)$ or outside of $R(a, b)$, and (ii) $d_R(b, c) \le d_R(a, b)$.
%\label{lem:outR}
%\end{lemma}
%\begin{proof}
%Rotate $S$ so that $b, c \in C_1(a)$.
%To prove property (i), assume to the contrary that $c$ lies strictly inside $R(a, b)$.
%Then $R(a, c) \subset R(a, b)$, since $b$ and $c$ are in the same cone of $a$
%(refer to Figure~\ref{fig:outR}).
%This implies that $d_R(a, c) < d_R(a, b)$, a contradiction. We now turn to property $2$.
%\end{proof}

Following is an intermediate result that will help prove our main upper bound result
stated in Theorem~\ref{thm:y3connect}.
This intermediate result will simply rule out some configurations that will occur in
the analysis of the main result. To follow the logical sequence of our analysis,
the reader can skip ahead to Theorem~\ref{thm:y3connect}, and refer back
to Lemma~\ref{lem:tri} only when called upon from Theorem~\ref{thm:y3connect}.

\begin{lemma}
Let $a, b, c \in S$ be such that $b, c \in C_i(a)$, for some $i \in \{1, 2, 3\}$,
and $|ac| \le |ab|$. Furthermore, assume that both $b$ and $c$ lie either in the
half of $C_i(a)$ adjacent to $r_i$ (excluding the bisector points), or in the
half of $C_i(a)$ adjacent to $r_{i+1}$ (including the bisector points).
Then $d_R(b,c) < d_R(a,b)$.
\label{lem:tri}
\end{lemma}
\begin{proof}
To simplify our analysis, rotate $S$ so that both $b$ and $c$ lie in the lower half of
$C_1(a)$ (adjacent to $r_1$). Let $\delta = d_R(a, b)$.
By Lemma~\ref{lem:dineq}(ii), $|ab| \le \delta$. This along with $|ac| \le |ab|$ implies that
$c \in D(a, \delta)$. More precisely, $c$ lies in a circular sector of angle $60^\circ$,
%-- call it $D_{60^\circ}(a)$ --
formed by the intersection between $D(a, \delta)$ and the lower half of $C_1(a)$.
%See Figure~\ref{fig:trilemma}.

%%%%%%%%%%%%%%%%%%%%%%%%%%%%%%%%%Figure Begin
\begin{figure}[hptb]
\centering
\includegraphics[width=0.5\linewidth]{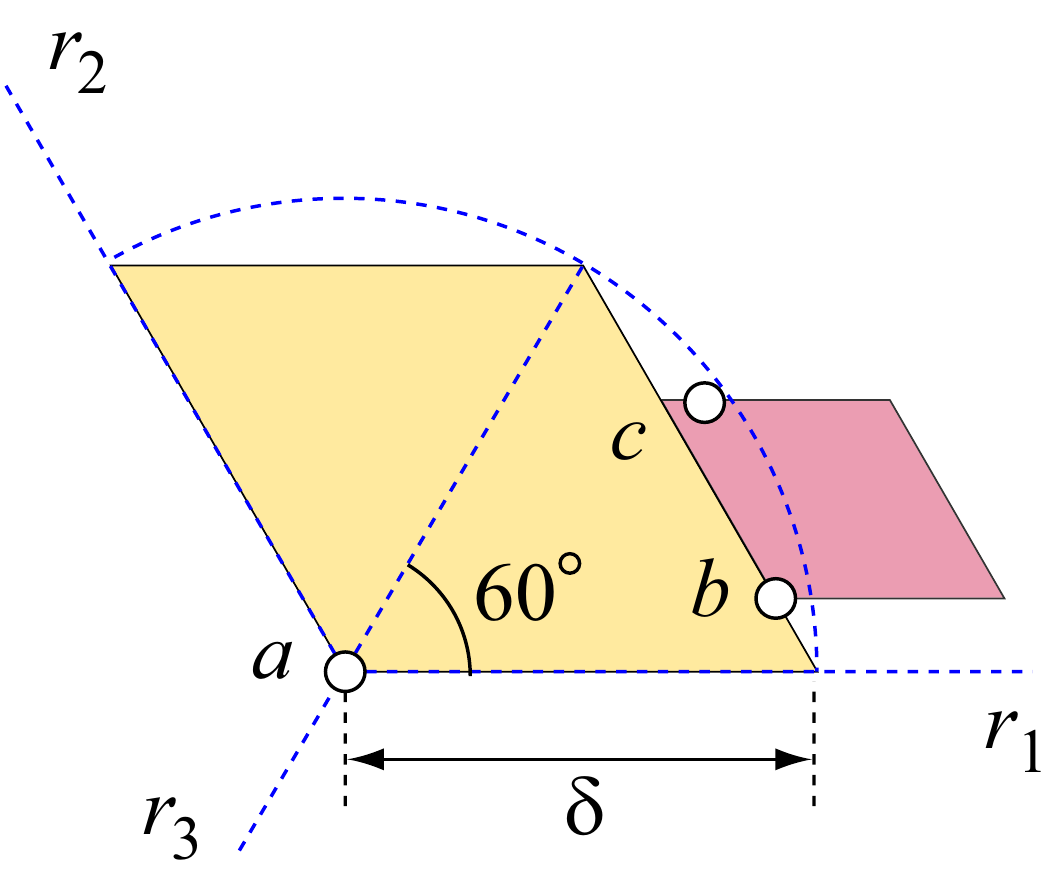}
\caption{Lemma~\ref{lem:tri}, case $c \in C_1(b)$:  $d_R(b, c) \le d_R(a, b)$.}
\label{fig:trilemma1}
\end{figure}
%%%%%%%%%%%%%%%%%%%%%%%%%%%%%%%%%Figure End

If $c \in C_1(b)$, then $R(a, b)$ and $R(b, c)$ are similar (see Figure~\ref{fig:trilemma1}).
This along with the fact that $c$ lies in a same $60^\circ-$sector as $b$ implies that
$d_R(b,c) < d_R(a,b)$.
(The inequality is strict due to the fact that the lower half of the cone $C_1(a)$ does not
include the upper bounding ray.)
%%%%%%%%%%%%%%%%%%%%%%%%%%%%%%%%%Figure Begin
\begin{figure}[hptb]
\centering
\includegraphics[width=\linewidth]{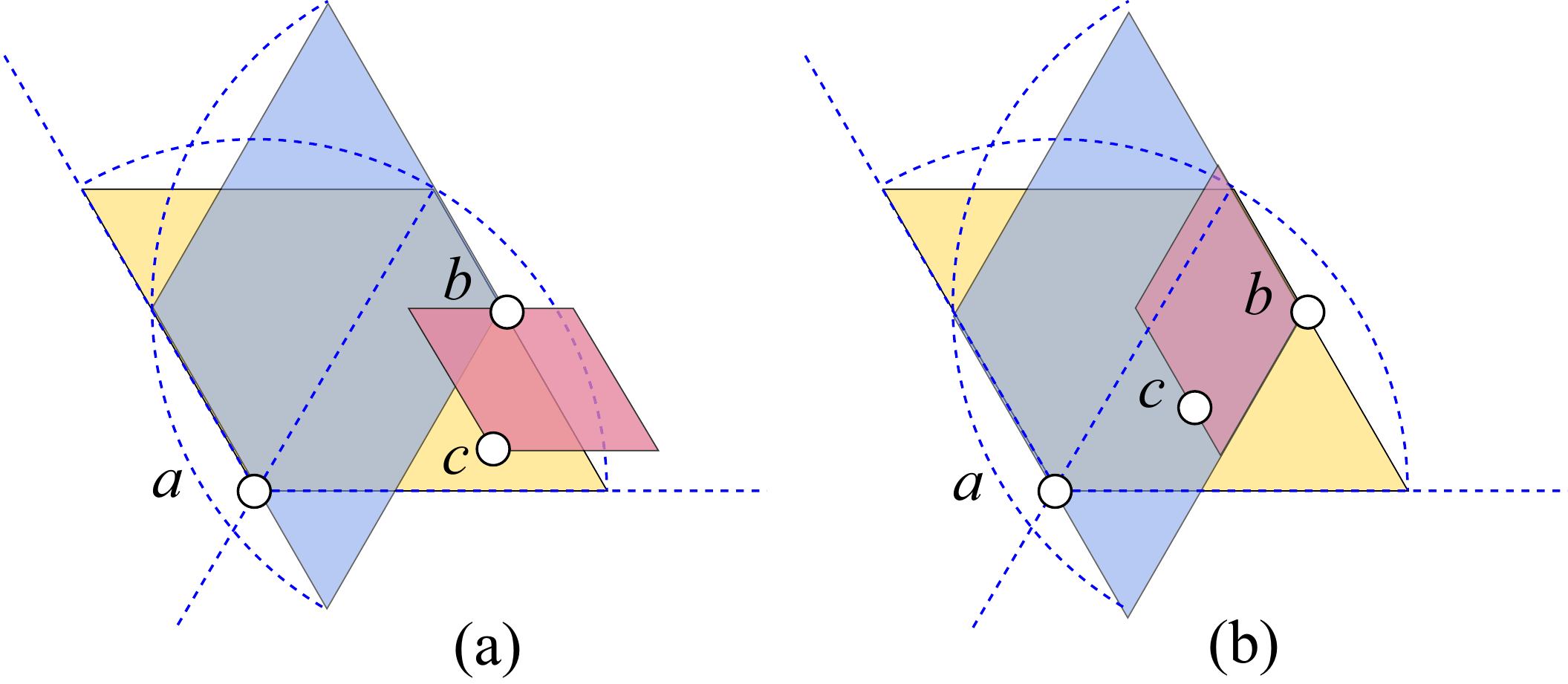}
\caption{Lemma~\ref{lem:tri}: $d_R(b, c) \le d_R(a, b)$ (a) $c \in C_3(b)$ (b) $c \in C_2(b)$.}
\label{fig:trilemma}
\end{figure}
%%%%%%%%%%%%%%%%%%%%%%%%%%%%%%%%%Figure End
If $c \in C_3(b)$, then $b \in C_1(c)$ (see Figure~\ref{fig:trilemma}a).
This case is similar to the previous one:
$R(a, b)$ and $R(c, b)$ are similar, and $d_R(b,c) = d_R(c, b) < d_R(a,b)$.
Finally, if $c \in C_2(b)$, then
$c \in R(b, a)$ (see Figure~\ref{fig:trilemma}b), and $R(b, c) \subset R(b, a)$.
It follows that $d_R(b,c) < d_R(b,a) = d_R(a, b)$.
\end{proof}

\begin{theorem}
For any point set $S$ such that $G^1(S)$ is connected, $Y_3[G^d]$ is also
connected, for $d = \frac{2}{\sqrt{3}}$.
\label{thm:y3connect}
\end{theorem}
\begin{proof}
The proof is by contradiction.
Assume to the contrary that $G^1$ is connected, but
$Y_3[G^d]$ is disconnected. Then $Y_3[G^d]$ has at least two connected
components, say $J_1$ and $J_2$. Since $G^1 \subseteq G^d$ is connected,
there is an edge $pq \in G^1$, with $p \in J_1$ and $q \in J_2$.
To derive a contradiction, consider two points $a, b \in S$, with
$a \in J_1$ and $b \in J_2$, that minimize $d_R(a, b)$. Then
$d_R(a, b) \le d_R(p, q) \le d\cdot|pq|$. This latter
inequality follows from inequality (iii) of Lemma~\ref{lem:dineq}, and
the $d$ value from the lemma statement. This along with
inequality (ii) of Lemma~\ref{lem:dineq} and the fact that
$|pq| \le 1$, implies that $|ab| \le d_R(a, b) \le d$,
therefore $ab \in G^{d}$.

%%%%%%%%%%%%%%%%%%%%%%%%%%%%%%%%%Figure Begin
\begin{figure}[htpb]
\centering
\includegraphics[width=0.5\linewidth]{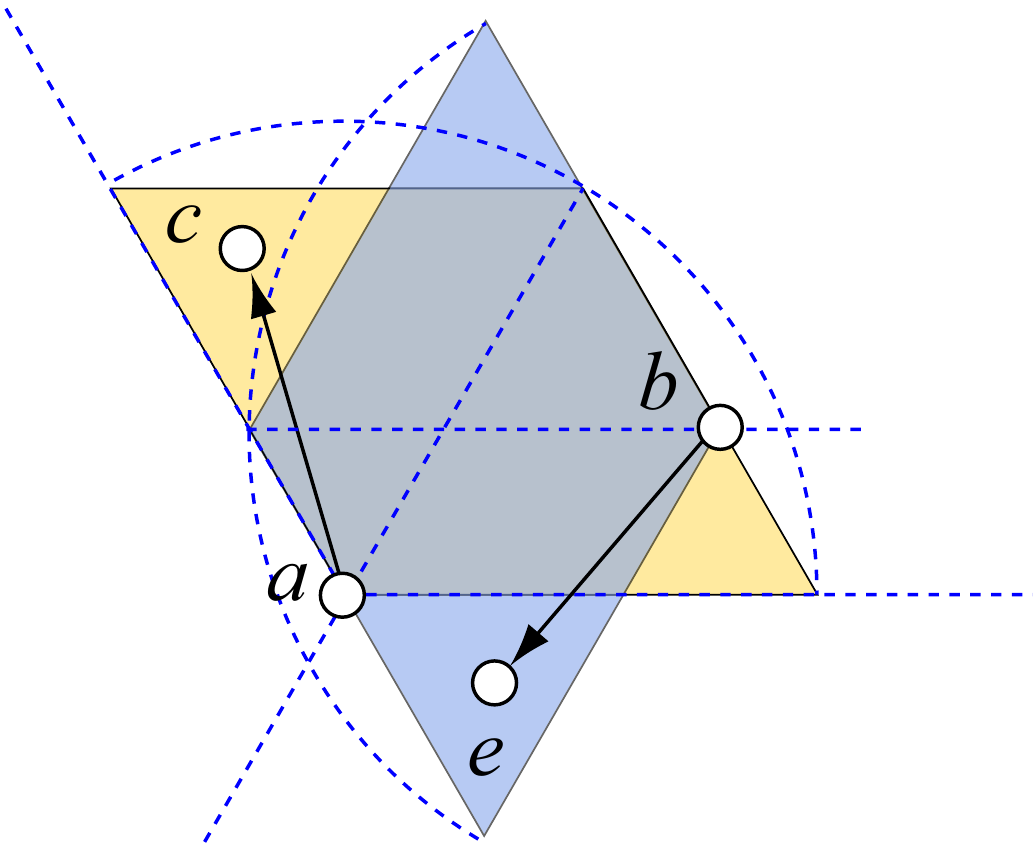}
\caption{Theorem~\ref{thm:y3connect}: case when $a$ and $d$ lie on
a same side of the bisector of $C_2(b)$.}
\label{fig:y3connect1}
\end{figure}
%%%%%%%%%%%%%%%%%%%%%%%%%%%%%%%%%Figure End

To simplify our analysis, rotate $S$ so that $b \in C_1(a)$.
Because $J_1$ and $J_2$ are not connected in $Y_3[G^{d}]$,
and because $a \in J_1$ and $b \in J_2$, we have that $ab \not\in Y_3[G^{d}]$.
However $ab \in G^{d}$ and $b \in C_1(a)$, therefore
there is $\arr{ac} \in Y_3[G^{d}]$, with $c \in C_1(a)$ and
$|ac| \le |ab|$. If both $b$ and $c$ lie in the same half of $C_1(a)$
(bounded by one ray and the bisector of $C_1(a)$),
then by Lemma~\ref{lem:tri} we have that $d_R(b, c) < d_R(a, b)$.
This along with the fact that $bc$ connects $J_1$ and $J_2$,
contradicts our choice of $ab$. Then $b$ and $c$ must lie on
either side of the bisector of $C_1(a)$, as depicted in
Figure~\ref{fig:y3connect1}.

%%%%%%%%%%%%%%%%%%%%%%%%%%%%%%%%%Figure Begin
\begin{figure}[hptb]
\centering
\includegraphics[width=0.7\linewidth]{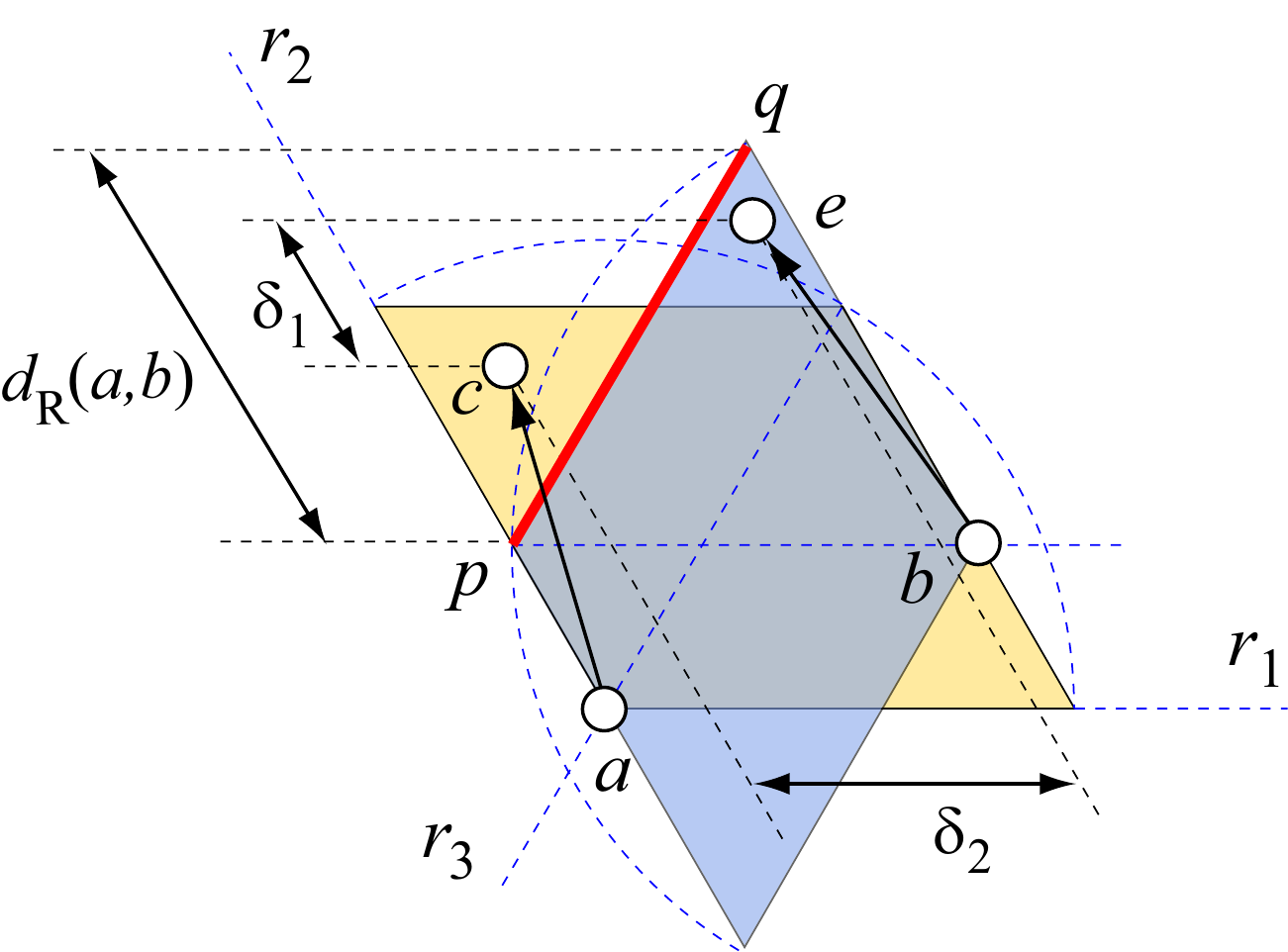}
\caption{Theorem~\ref{thm:y3connect}: case when $a$ and $d$ lie on
opposite sides of the bisector of $C_2(b)$.}
\label{fig:y3connect2}
\end{figure}
%%%%%%%%%%%%%%%%%%%%%%%%%%%%%%%%%Figure End

Assume without loss of generality that $b$ lies in the lower half
(excluding the bisector) of $C_1(a)$, and
$c$ lies in the upper half (including the bisector) of $C_1(a)$.
Next we focus on $C_2(b)$. Because $a \in C_2(b)$, $ba \in G^d$, and
$ba \not\in Y_3[G^{d}]$, there must exist $\arr{be} \in Y_3[G^{d}]$,
with $e \in C_2(b)$ and $|be| \le |ab|$. As before, if $e$
and $a$ lie in the same half of $C_2(b)$ (bounded by one ray and the
bisector of $C_2(b)$), then by Lemma~\ref{lem:tri} we have that
$d_R(e, a) < d_R(b, a) = d_R(a, b)$.
This along with the fact that $ae$ connects $J_1$ and $J_2$
contradicts our choice of $ab$. It follows that $a$ and $e$ lie on
either side of the bisector of $C_2(b)$, as depicted in
Figure~\ref{fig:y3connect2}.

We now show that $d_R(c, e) < d_R(a, b)$. Let $\delta_1$ be the
length of the projection of $ce$ on the ray $r_2$ in the (horizontal)
direction of $r_1$. Similarly, let $\delta_2$ be the length of the
projection of $ce$ on $r_1$ in the direction of $r_2$.
(See Figure~\ref{fig:y3connect2}.)
Then $d_R(c, e) = \max\{\delta_1, \delta_2\}$. We prove that
$d_R(c, e) < d_R(a, b)$ by showing that each of
$\delta_1$ and $\delta_2$ is smaller than $d_R(a, b)$.

%%%%%%%%%%%%%%%%%%%%%%%%%%%%%%%%%Figure Begin
\begin{figure}[hptb]
\centering
\includegraphics[width=\linewidth]{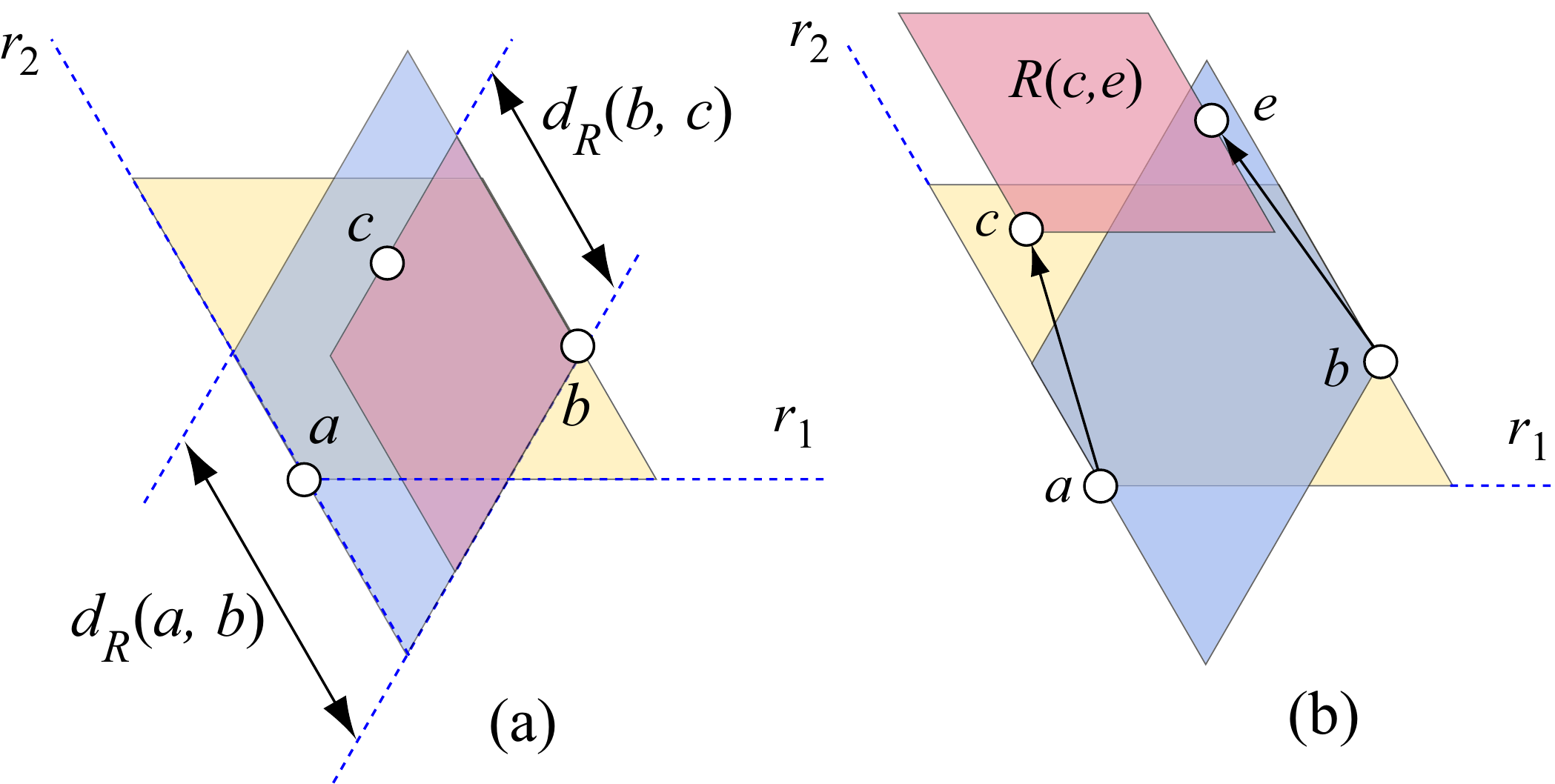}
\caption{Theorem~\ref{thm:y3connect}: (a) $c$ inside $R(b, a)$ (b) $d_R(c,e) < d_R(a, b)$.}
\label{fig:y3connect}
\end{figure}
%%%%%%%%%%%%%%%%%%%%%%%%%%%%%%%%%Figure End

First note that $c$ must lie outside of $R(b, a)$. Otherwise,
if $c$ were to lie inside $R(b, a)$, then $R(b, c) \subset R(b, a)$
(see Figure~\ref{fig:y3connect}a).
This would immediately imply that $d_R(b, c) < d_R(b, a) = d_R(a,b)$,
which along with the fact that $bc$ connects $J_1$ and $J_2$,
would contradict our choice of $ab$.
So $c$ lies inside $D(a, |ab|)$ (because $|ac| \le |ab|$),
but outside of $R(b, a)$.
Similar arguments show that $e$ lies inside $D(b, |ab|)$, but
outside of $R(a, b)$.
Let $pq$ be the top left side of $R(b, a)$
(marked with a thick line in Figure~\ref{fig:y3connect2}).
By the observations above, $c$ and $e$ cannot lie below $p$ or above $q$.
This implies that the horizontal projection of $ce$ on the ray $r_2$
is strictly shorter than the horizontal projection of $pq$ on $r_2$:
 $\delta_1 < d_R(a,b)$. (The claim on \emph{strictly} shorter comes from the
 fact that $c \in C_1(a, b)$, and $C_1(a, b)$ does not include $r_2$.)
Also, because $c$ and $e$ lie between the two lines through $a$ and $b$
parallel to $r_2$, the projection of $ce$ on $r_1$ in the direction of $r_2$
is strictly shorter than the projection of $ab$ on $r_1$ in the direction
of $r_2$: $\delta_2 < d_R(a, b)$.

We have established that $d_R(c, e) < d_R(a, b)$ (the rhombus $R(c, e)$ is
depicted in Figure~\ref{fig:y3connect}b). This along with the fact that
$ce$ connects $J_1$ and $J_2$, contradicts our choice of $ab$. We conclude
that $G^{d}$ is connected.
\end{proof}

\noindent
Observe that our results leave a tiny gap between the lower bound of
$5-\frac{2}{3}\sqrt{35} \approx 1.056$ from Theorem~\ref{thm:y3lb} and the
upper bound of $\frac{2}{\sqrt{3}} \approx 1.155$ from Theorem~\ref{thm:y3connect}
on the connectivity radius $d$ for $Y_3[G^d]$.
Nevertheless, both bounds beat the tight bound $d = \sqrt{2} \approx 1.414$
for the connectivity radius of $Y_4[G^d]$.

\section{Connectivity of $Y_2$}
The point set $S$ depicted in Figure~\ref{fig:y2disconnect} can be extended to
show that $Y_2[G^d]$ can be disconnected, for arbitrarily large $d$. 
To see this, fix a real value $d \ge 1$, and distribute enough points $a_i$
at unit interval along the leftward ray from $p$, such that the leftmost
point $a_r$ is far enough from $q$ -- in particular, we require that it satisfies
the inequality $|a_rq| > d$.  
Similarly, we require that the rightmost point point $b_r$ satisfies $|b_rp| > d$
(which follows immediately by symmetry). (Note that in this case $d = \Omega(|S|)$.) 
Recall that the leftward ray from $p$ is \emph{almost} horizontal, so $q$ and
all the $b-$points lie above $a_r$.

We now show that $Y_2[G^d]$ is disconnected. First note that $a_1$ is the
point closest to $p$ in $C_1(p)$, and that $C_2(p)$ is empty. Therefore, the
only edge $Y_2[G^d]$ incident to $p$ is $pa_1$. Also note that, for
any $i < r$,  $a_{i+1}$ is the point closest to $a_i$ in $C_1(a_i)$, and
$a_{i-1}$ is the point closest to $a_i$ in $C_2(a_i)$ (here we use $a_0$ to
refer to the point $p$). Finally, $q$ is the point closest to $a_r$ in
$C_1(a_r)$. However, because $|a_rq| > d$, $a_rq$ is not in $G^d$ and
therefore $a_rq$ is not in $Y_2[G^d]$.
The arguments are symmetric for $q$ and the $b-$points in $S$.
This shows that there is no edge in $Y_2[G^d]$ connecting a point in
$\{q, b_i ~|~ 1 \le i \le r\}$ to a point in $\{p, a_i ~|~ 1 \le i \le r\}$.
We conclude that $Y_2[G^d]$ is disconnected for connectivity radius values 
$d = \Omega(|S|)$.

\section{Conclusion}
In this paper we establish matching lower and upper bounds on the connectivity
radius for $Y_4$, and a tight interval on the connectivity radius for $Y_3$. Reducing
the gap between the lower and upper ends of this interval remains open.
These results show that a small
increase in the radius of a directional antenna, (compared to the unit
radius of an omnidirectional antenna,) renders an efficient communication
graph for mobile wireless networks, provided that each node orients its
$k \in \{3, 4\}$ antennas in the direction of the $Y_k$ edges. (Nodes are assumed
to send messages in directional mode, and receive messages in
omnidirectional mode). One key advantage of these graphs is that they
can be quickly constructed locally, providing strong support for node mobility.
We also establish that the connectivity radius for $Y_2$ may be arbitrarily
large, which indicates  that $Y_2$ is not a suitable communication graph for
wireless networks that use narrow (laser-beam) directional antennas.

\def\cprime{$'$}

%\newpage
%\medskip
%\bibliographystyle{plain}
%\bibliography{Spanners}

\end{document}